%% file: sirocco-main.tex
\RequirePackage{amsmath}

\documentclass[runningheads,a4paper]{llncs}

\usepackage[lmargin=2.5cm,rmargin=2.5cm,bottom=2.5cm,top=2.5cm,twoside=False]{geometry}

\usepackage[defaultlines=3,all]{nowidow}
\usepackage{ntheorem}
\usepackage{mathtools}
\usepackage{graphicx}
\usepackage{mathrsfs}
\usepackage{xcolor}
\usepackage{amssymb}
\usepackage{mathabx}
\usepackage{tikz}
\definecolor{cadmiumgreen}{rgb}{0.0, 0.42, 0.24}
\definecolor{dark-blue}{rgb}{0.05,0.25,1}
\usepackage[colorlinks=true,linkcolor={blue},citecolor=blue]{hyperref}
\usepackage[english]{babel}
\usepackage{microtype}
\usepackage{tcolorbox}
\usepackage[noabbrev,capitalise,nameinlink]{cleveref}
\newcommand\claimref[2]{\hyperref[{#1}]{#2}}

\newcommand{\nc}[1]{\newcommand{#1}}

\newcommand{\minor}{\preceq}

\nc{\e}{\epsilon}
\nc{\dd}{\delta}
\nc{\dom}{dom}
\nc{\nn}{\nabla}
\nc{\bn}{\nabla\!\!\!\nabla_1}

\usepackage{pgfplots}
\pgfplotsset{compat = newest}

 	\newtheorem{cor}[lemma]{Corollary}

\newtheorem{alemma}[lemma]{Adapted Lemma}
\newtheorem{acor}[lemma]{Adapted Corollary}

\crefformat{alemma}{#2Adapted Lemma~#1#3}
\Crefformat{acor}{#2Adapted Corollary~#1#3}

\begin{document}
\title{Local planar domination revisited}
%
%\titlerunning{Abbreviated paper title}
% If the paper title is too long for the running head, you can set
% an abbreviated paper title here
%
\author{Ozan Heydt\inst{1} \and
Sebastian Siebertz\inst{1} \and% \\%\linebreak
Alexandre Vigny\inst{1}}
\authorrunning{O.\ Heydt, S.\ Siebertz and A.\ Vigny}
% First names are abbreviated in the running head.
% If there are more than two authors, 'et al.' is used.
%
\institute{University of Bremen, Germany\\
\email{\{heydt,siebertz,vigny\}@uni-bremen.de}}
\maketitle              % typeset the header of the contribution
%

\input{0-abstract}

\input{1-intro}

\input{2-prelim}

\input{3-preprocessing}

\input{4-high-degree}

\input{5-res-degree}

\input{6-girth}

\input{8-conclusion}

%
% ---- Bibliography ----
%
% BibTeX users should specify bibliography style 'splncs04'.
% References will then be sorted and formatted in the correct style.
%
\pagebreak
 \bibliographystyle{splncs04}
 \bibliography{ref}
 
 \pagebreak
 \appendix 
 
 \input{7-cplex}
\end{document}

%% file: 0-abstract.tex
% !TEX root = sirocco-main.tex

\begin{abstract}
  We show how to compute a 20\hspace{1pt}-\hspace{1pt}approximation of
  a minimum dominating set in a planar graph in a constant number of
  rounds in the LOCAL model of distributed computing. This improves on
  the previously best known approximation factor of 52, which was
  achieved by an elegant and simple algorithm of Lenzen et al.\ Our
  algorithm combines ideas from the algorithm of Lenzen et al.\ with
  recent work of Czygrinow et al.\ and Kublenz et al.\ to reduce to
  the case of bounded degree graphs, where we can simulate a
  distributed version of the classical greedy algorithm.
  \keywords{Dominating set \and LOCAL algorithms \and Planar graphs}
\end{abstract}

%% file: 1-intro.tex
% !TEX root = sirocco-main.tex

\section{Introduction}\label{sec:intro}

A dominating set in an undirected and simple graph $G$ is a set
$D\subseteq V(G)$ such that every vertex $v\in V(G)$ either belongs
to $D$ or has a neighbor in $D$. The dominating set problem is a
classical NP-complete problem~\cite{karp1972reducibility} with many
applications in theory and practice, see e.g.~\cite{du2012connected,sasireka2014applications}. In this paper we study
the distributed time complexity of finding
dominating sets in planar graphs in the classical LOCAL model of distributed computing.
In this model, a distributed system is modeled
by an undirected (planar) graph~$G$. Every vertex represents a
computational entity and the vertices communicate through the edges of
$G$. The vertices are equipped with unique identifiers and initially,
every vertex is only aware of its own identity. A computation then
proceeds in synchronous rounds.  In every round, every vertex sends
messages to its neighbors, receives messages from its neighbors and
performs an arbitrary computation.  The complexity of a LOCAL
algorithm is the number of rounds until all vertices return their
answer, in our case, whether they belong to a dominating set
or~not.

The problem of approximating
dominating sets in the LOCAL model has received considerable
attention in the literature. Since in general graphs
it is not possible to compute a
constant factor approximation in a constant number of rounds~\cite{KuhnMW16},
much effort has been invested to improve the ratio between approximation
factor and number of rounds on special graph classes. A very successful
line of structural analysis of graph properties that can lead to improved
algorithms was started by the influential paper of Lenzen et al.~\cite{lenzen2013distributed}, who in particular proved that on planar graphs
a 130-approximation
of a minimum dominating set can be computed in a constant number of
rounds. A careful analysis of Wawrzyniak~\cite{wawrzyniak2014strengthened}
later showed that the algorithm computes in fact a 52-approximation.
In terms of lower bounds, Hilke et al.~\cite{hilke2014brief} showed that there is no
deterministic local algorithm (constant-time distributed graph algorithm) that
finds a~$(7-\epsilon)$-approximation of a minimum dominating set on
planar graphs, for any positive constant~$\epsilon$. Better approximation
ratios are known for some special cases, e.g.\ 32 if the planar graph is
triangle-free \mbox{\cite[Theorem 2.1]{alipour2020distributed}}, 18 if the planar graph has girth
five~\cite{alipour2020local} and 5 if the graph is
outerplanar (and this bound is tight)~\cite[Theorem 1]{bonamy2021tight}.

In this work we tighten the gap between the best-known
lower bound of~$7$ and the best-known upper bound of $52$ on planar graphs
by providing a new approximation algorithm computing a 20\hspace{1pt}-\hspace{1pt}approximation.

Our algorithm proceeds in three phases.
The first phase is a preprocessing phase that was similarly employed in the
algorithm of Lenzen et al.~\cite{lenzen2013distributed}. In a key lemma,
Lenzen et al.\ proved that there are only few vertices whose open neighborhood
cannot be dominated by at most six vertices. We improve this lemma and show
that there are only slightly more vertices whose open neighborhood cannot be
dominated by \emph{three} other vertices. All these vertices are selected into an
initial partial dominating set and as a consequence the open neighborhoods of all
remaining vertices can be dominated be at most three vertices.

By defining the notion of \emph{pseudo-covers},
Czygrinow et  al.~\cite{czygrinow2018distributed} provided a tool to carry
out a fine grained analysis of the vertices that can potentially dominate
the remaining neighborhoods.
Using ideas of~\cite{kublenz2020distributed} and~\cite{siebertz2019greedy}
we provide an even finer analysis for planar graphs on which we
base the second phase of our distributed algorithm and compute a
second partial dominating set.

\pagebreak

We prove that
after the second phase we are left with a graph where every vertex
has at most 30 non-dominated neighbors. In particular, every vertex
from a minimum dominating set $D$ can dominate at most 30
non-dominated vertices, hence, we could at this point pick all
non-dominated vertices to add at most $31|D|$ vertices (each
vertex dominates its neighbors and itself). We could
also apply a general algorithm of Lenzen and Wattenhofer
that computes in a graph of arboricity $a$ and
maximum degree~$\Delta$ a \mbox{$16a\log \Delta$-approximation}
in $6\left\lceil \log \Delta+1\right\rceil$ rounds~\cite{lenzen2010minimum}.
Planar graphs have arboricity $3$ and $\log 30\approx 4.907$, hence,
%\alex{Are we sure it is $\log_2$?\\ If so, then replace $\log 114\approx 6.833$
%by $\log 30\approx 4.907$, so $16a\log 30\approx 235$}
in our situation $16a\log \Delta\approx 235$ and this would not yield an improvement.
Of course, Lenzen and Wattenhofer
optimized not only towards minimizing the approximation factor, which
they could have easily improved, but
also towards minimizing the number of rounds with respect to~$\Delta$. This is well-motivated as in general graphs the maximum
degree can be large, however, in our algorithm
we always arrive at this fixed constant
degree so we can now proceed in a constant number of rounds.

We proceed in a greedy manner in $30$ rounds as follows.
Call the number of non-dominated neighbors of a vertex $v$ the
\emph{residual degree} of $v$. In the first round, if a non-dominated
vertex has a neighbor of residual degree~$30$, it elects one such neighbor into
the dominating set (or if it has residual degree $30$ itself, it may choose itself).
The neighbors of the chosen elements are
marked as dominated and the residual degrees are updated. Note that
all non-dominated neighbors of a vertex of residual degree $30$
in this round choose a dominator, hence, the residual degrees
of all vertices of residual degree $30$ are decreased to $0$, hence, after
this round there are no vertices of residual degree $30$ left.
In the second round, if a non-dominated vertex has a neighbor
of residual degree~$29$, it elects
one such vertex into the dominating set, and so on, until after $30$ rounds
in the final round every vertex chooses a dominator. Unlike in the general
case, where
nodes cannot learn the current maximum residual degree in a constant
number of rounds, by establishing
an upper bound on the maximum residual degree and proceeding in exactly
this number of rounds, we ensure that we iteratively exactly choose the
vertices of maximum residual degree. It remains to analyze the performance
of this algorithm.

A simple density argument shows
that there cannot be too many vertices of degree $i\geq 6$ in a planar
graph. At a first glance it seems that the algorithm would perform worst
when in every of the $30$ rounds it would pick as many vertices as possible,
as the constructed dominating set would grow as much as possible. However,
this is not the case, as picking many high degree vertices at the same time makes
the largest progress towards dominating the whole graph. It turns
out that there is a delicate balance between the vertices that we pick
in round $i$ and the remaining non-dominated vertices that leads
to the worst case. We formulate these conditions as a
linear program and solve the linear program. In total, this
leads to the claimed 20\hspace{1pt}-\hspace{1pt}approximation (\cref{thm:planar}).

\smallskip
We then analyze our algorithm on more restricted graphs classes, and prove that
it computes approximations of factors: 14 for triangle-free planar graphs,
13 for bipartite planar graphs, 7 for planar graphs of girth $5$, and
12 for outerplanar graphs (\cref{thm:bip,thm:tri,thm:girth,thm:outer}).
% We then analyze our algorithm on bipartite planar graphs,
% triangle-free planar graphs,
% planar graphs of girth $5$ and outerplanar graphs and prove that
% it computes a 12, 14, 7 and 12 approximation, respectively (\cref{thm:bip,thm:tri,thm:girth,thm:outer}).
This
improves the currently best known approximation ratios of 32
and 18 for triangle-free planar graphs and planar graphs of girth $5$,
respectively, while our algorithm fails short of achieving
the optimal approximation ratio of 5 on outerplanar graphs.

%While we could tighten the gap between the best known lower and upper
%bound on planar graphs there is still much room for improvement. We
%believe that the optimum approximation rate is much closer to $7$ than
%to $24$. We also stress that we are not satisfied with our brute-force
%CPLEX proof, but we must unfortunately leave a more elegant proof
%for future work.

%
%, see e.g.~\cite{akhoondian2018distributed,%
%akhoondian2016local,%
%alipour2020local,%
%amiri2016brief,%
%amiri2019distributed,%
%barenboim2018fast,%
%czygrinow2008fast,%
%czygrinow2018distributed,%
%DBLP:conf/stoc/GhaffariKM17,%
%hilke2014brief,%
%kublenz2020distributed,%
%KuhnMW16,%
%lenzen2013distributed,%
%lenzen2008leveraging,%
%lenzen2010minimum,%
%DBLP:conf/stoc/RozhonG20,%
%wawrzyniak2013brief,%
%wawrzyniak2014strengthened}.

%% file: 2-prelim.tex
% !TEX root = sirocco-main.tex

\section{Preliminaries}

In this paper we study the distributed time complexity of finding
dominating sets in planar graphs in the classical LOCAL model of
distributed computing.
We assume familiarity with this model and refer to the survey~\cite{suomela2013survey} for extensive
background on distributed computing and the LOCAL model.

We use standard notation from graph theory and refer to the
textbook~\cite{diestel} for extensive background.  All graphs in this
paper are undirected and simple. We write $V(G)$ for the vertex set of
a graph $G$ and~$E(G)$ for its edge set. The \emph{girth} of a graph
$G$ is the length of a shortest cycle in $G$. A graph is called
\emph{triangle-free} if it does not contain a triangle, that is, a
cycle of length three as a subgraph. Equivalently, a triangle-free
graph is a graph of girth at least four.

A graph is \emph{bipartite} if its vertex set can be partitioned into
two parts such that all its edges are incident with two vertices from
different parts. More generally, the \emph{chromatic number}~$\chi(G)$
of a graph $G$ is the minimum number $k$ such that the vertices of $G$
can be partitioned into $k$ parts such that all edges are incident
with two vertices from different parts. Hence, the bipartite graphs
are exactly the graphs with chromatic number two. A set $A$ is
\emph{independent} if all two distinct vertices $u,v\in A$ are
non-adjacent.  Every graph $G$ contains an independent set of size at
least $\lceil|V(G)|/\chi(G)\rceil$.

\pagebreak
A graph is \emph{planar} if it can be embedded in the plane, that is,
it can be drawn on the plane in such a way that its edges intersect
only at their endpoints. By the famous theorem of Wagner, planar
graphs can be characterized as those graphs that exclude the complete
graph $K_5$ on five vertices and the complete bipartite $K_{3,3}$ with
parts of size three as a minor.  A graph~$H$ is a \emph{minor} of a
graph~$G$, written~$H\minor G$, if there is a set
\mbox{$\{G_v :v\in V(H)\}$} of pairwise vertex disjoint and connected
subgraphs $G_v\subseteq G$ such that if~$\{u,v\}\in E(H)$, then there
is an edge between a vertex of~$G_u$ and a vertex of~$G_v$. We
call~$V(G_v)$ the \emph{branch set} of $v$ and say that it is
\emph{contracted} to the vertex~$v$.  We call $H$ a \emph{$1$-shallow
  minor}, written~$H\minor_1 G$, if $H\minor G$ and there is a minor
model \mbox{$\{G_v :v\in V(H)\}$} witnessing this, such that all
branch sets $G_v$ have radius at most $1$, that is, in each $G_v$
there exists $w$ adjacent to all other vertices of $G_v$. In other
words, $H\minor_1 G$ if $H$ is obtained from $G$ by deleting some
vertices and edges and then contracting a set of pairwise disjoint
stars. We refer to~\cite{nevsetvril2012sparsity} for an in-depth study
of the theory of sparsity based on shallow minors.

A graph is \emph{outerplanar} if it has an embedding in the plane such
that all vertices belong to the unbounded face of the embedding.
Equivalently, a graph is outerplanar if it does not contain the
complete graph $K_4$ on four vertices and the complete bipartite graph
$K_{2,3}$ with parts of size $2$ and $3$, respectively, as a minor. If
$J\minor H$ and $H\minor G$, then $J\minor G$, hence a minor of a
planar graph is again planar and a minor of an outerplanar graph is
again outerplanar.

By Euler's formula, planar graphs are sparse: every planar $n$-vertex
graph ($n\geq 3$) has at most $3n-6$ edges (and a graph with at most
two vertices has at most one edge). The ratio $|E(G)|/|V(G)|$ is
called the \emph{edge density} of $G$. In particular, every planar
graph $G$ has edge density strictly smaller than three.

\begin{lemma}\label{lem:densities}
  Let $G$ be a planar graph. Then the edge density of $G$ is strictly
  smaller than $3$ and $\chi(G)\leq 4$. Furthermore,

  \vspace{-1mm}
  \begin{enumerate}
  \item if $G$ is bipartite, then the edge density of $G$ is strictly
    smaller than $2$ and $\chi(G)\leq 2$,\smallskip
  \item if $G$ is triangle-free or outerplanar, then the edge density
    of $G$ is strictly smaller than $2$ and $\chi(G)\leq 3$.
    % and\smallskip
    % \item if $G$ has girth at least $5$, then the edge density of
    %   $G$ is strictly smaller than $\frac{5}{3}$ and
    %   $\chi(G)\leq 3$.
  \end{enumerate}
\end{lemma}

For a graph $G$ and $v\in V(G)$ we write
$N(v)=\{u~:~\{u,v\}\in E(G)\}$ for the \emph{open neighborhood} of $v$
and $N[v]=N(v)\cup\{v\}$ for the \emph{closed neighborhood}
of~$v$. For a set $A\subseteq V(G)$ let $N[A]=\bigcup_{v\in A}N[v]$.
% We write $N_r[v]$ for the set of vertices at distance at most $r$
% from a vertex $v$.
A~\emph{dominating set} in a graph~$G$ is a set $D\subseteq V(G)$ such
that $N[D]=V(G)$. We write $\gamma(G)$ for the size of a minimum
dominating set of $G$. For $W\subseteq V(G)$ we say that a set
$Z\subseteq V(G)$ \emph{dominates} $W$ if $W\subseteq N[Z]$.

\smallskip
In the following we mark important definitions and assumptions about
our input graph in gray boxes and steps of the algorithm in red boxes.

\begin{tcolorbox}
  We fix a planar graph $G$ and a minimum dominating set~$D$ of $G$
  with $\gamma \coloneqq |D|=\gamma(G)$.
\end{tcolorbox}

%% file: 3-preprocessing.tex
% !TEX root = sirocco-main.tex

\section{Preprocessing}\label{sec:step1}

As outlined in the introduction, our algorithm works in three phases.
In phase~$i$ for $1\leq i\leq 3$ we select a partial dominating set
$D_i$ and estimate its size in comparison to $D$. In the end we will
return $D_1\cup D_2\cup D_3$. We will call vertices have been selected
into a set $D_i$ \emph{green}, vertices that are dominated by a green
vertex but are not green themselves are called \emph{yellow} and all
vertices that still need to be dominated are called \emph{red}. In the
beginning, all vertices are marked red.

The first phase of our algorithm is similar to the first phase of the
algorithm of Lenzen et al.~\cite{lenzen2013distributed}. It is a
preprocessing step that leaves us with only vertices whose
neighborhoods can be dominated by a few other vertices. Lenzen et al.\
proved that there exist less than $3\gamma$
many vertices $v$ such
that the open neighborhood $N(v)$ of $v$ cannot be dominated by $6$
vertices of $V(G)\setminus \{v\}$~\cite[Lemma
6.3]{lenzen2013distributed}. The lemma can be generalized to more
general graphs, see~\cite{amiri2019distributed}). We prove the
following lemma, which is stronger in the sense that the number of
vertices required to dominate the open neighborhoods is smaller than
$6$, at the cost of having slightly more vertices with that property.

\smallskip We define $D_1$ as the set of all vertices whose
neighborhood cannot be dominated by $3$ other vertices.
\begin{tcolorbox}[colback=red!5!white,colframe=red!50!black]
\vspace{-4mm}
  \begin{align*}
  D_1\coloneqq \{v\in V(G) : \text{ for all sets } A\subseteq V(G)\setminus \{v\}
  & \text{ with $N(v)\subseteq N[A]$ we have $|A|> 3\}$.}
  \end{align*}
\end{tcolorbox}

\smallskip

We prove a very general lemma that can be applied also for more
general graph classes, even though we will apply it only for planar
graphs.  Hence, in the following lemma, $G$ can be an arbitrary graph,
while in the following lemmas $G$ will again be the planar graph that
we fixed in the beginning.

\begin{lemma}\label{lenzen-improved}
  Let $G$ be a graph, let $D$ be a minimum dominating set of $G$ of
  size $\gamma$ and let $\nn$ be an integer strictly larger than the
  edge density of a densest bipartite $1$-shallow minor of $G$. Let
  $\hat{D}$ be the set of vertices $v\in V(G)$ whose neighborhood
  cannot be dominated by $(2\nn-1)$ vertices of $D$ other than $v$,
  that is,
  \[
    \text{ $\hat{D}\coloneqq \{v\in V(G) :$ for all sets
      $A\subseteq D\setminus \{v\}$ with $N(v)\subseteq N[A]$ we have
      $|A|> (2\nn-1)\}$.}
  \]
  Then $|\hat{D}\setminus D| < \chi(G)\cdot\gamma$.
\end{lemma}

Recall that minors of planar graphs are again planar, hence, the
maximum edge density of a bipartite $1$-shallow minor of a planar
graph is smaller than $2$ and hence we can choose $\nn=2$ for the case
of planar graphs and we note the following corollary.
\begin{cor}\label{lenzen-improved-planar}\hspace{-1.7mm}\textbf{.}
  Let $\hat{D}$ be the set of vertices $v$ whose neighborhood cannot
  be dominated by $3$ vertices of $D$ other than $v$, that is,
  \[
    \text{ $\hat{D}\coloneqq \{v\in V(G) :$ for all sets
      $A\subseteq D\setminus \{v\}$ with $N(v)\subseteq N[A]$ we have
      $|A|>3\}$.}
  \]
  Then $|\hat{D}\setminus D| < 4\gamma$.
\end{cor}

\begin{proof}[of~\cref{lenzen-improved}]
  Assume $D=\{b_1,\ldots,b_\gamma\}$.  Assume that there are
  $\chi(G)\cdot\gamma$ vertices
  $a_1,\ldots,a_{\chi(G)\gamma}\not\in D$ satisfying the above
  condition. As the chromatic number is monotone over subgraphs, the
  subgraph induced by the $a_is$ is also $\chi(G)$-chromatic, so we
  find an independent subset of the $a_is$ of size $\gamma$. We can
  hence assume that $a_1,\ldots,a_{\gamma}$ are not connected by an
  edge. We proceed towards a contradiction.

  We construct a bipartite $1$-shallow minor $H$ of $G$ with the
  following \mbox{$2\gamma$} branch sets. For every
  \mbox{$i\le \gamma$} we have a branch set $A_i=\{a_i\}$ and a branch
  set
  \mbox{$B_i=N[b_i]\setminus (\{a_1,\ldots, a_{\gamma}\}\cup
    \bigcup_{j<i}N[b_j]\cup \{b_{i+1},\ldots, b_\gamma\})$}. Note that
  the $B_i$ are vertex disjoint and hence we define proper branch
  sets. Intuitively, for each vertex $v\in N(a_i)$ we mark the
  smallest $b_j$ that dominates $v$ as its dominator. We then contract
  the vertices that mark $b_j$ as a dominator together with $b_j$ into
  a single vertex. Note that because the $a_i$ are independent, the
  vertices $a_i$ themselves are not associated to a dominator as no
  $a_j$ lies in $N(a_i)$ for~$i\neq j$.  Denote by
  $a_1',\ldots, a_{\gamma}',b_1',\ldots, b_\gamma'$ the associated
  vertices of $H$. Denote by $A$ the set of the $a_i's$ and by $B$ the
  set of the~$b_j's$.  We delete all edges between vertices of
  $B$. The vertices of $A$ are independent by construction. Hence, $H$
  is a bipartite $1$-shallow minor of $G$.  By the assumption that
  $N(a_i)$ cannot be dominated by $2\nabla-1$ elements of $D$, we
  associate at least $2\nn$ different dominators with the vertices of
  $N(a_i)$. Note that this would not necessarily be true if $A$ was
  not an independent set, as all $a_j\in N(a_i)$ would not be
  associated a dominator.

  Since $\{b_1,\ldots, b_\gamma\}$ is a dominating set of $G$ and by
  assumption on $N(a_i)$, we have that in $H$, every~$a'_i$ has at
  least $2\nn$ neighbors in $B$. Hence,
  $|E(H)| \ge 2\nn|V(A)| = 2\nn\gamma$. As $|V(H)|=2\gamma$ we
  conclude $|E(H)|\ge \nn|V(H)|$. This however is a contradiction,
  as~$\nn$ is strictly larger than the edge density of a densest
  bipartite $1$-shallow minor of $G$.
\end{proof}

\bigskip

Let us fix the set $\hat{D}$ for our graph $G$.
\begin{tcolorbox}
\vspace{-4mm}
    \begin{align*}
      \hat D\coloneqq \{v\in V(G) : \text{ for all sets } A\subseteq D\setminus \{v\}
      & \text{ with $N(v)\subseteq N[A]$} \text{ we have $|A|>3\}$.}
  \end{align*}
\end{tcolorbox}
\smallskip

Note that $\hat{D}$ cannot be computed by a local algorithm as we do
not know the set $D$. It will only serve as an auxiliary set in our
analysis.

\smallskip The first phase of the algorithm is to compute the set
$D_1$, which can be done in 2 rounds of communication. Obviously, if
the open neighborhood of a vertex $v$ cannot be dominated by $3$
vertices from $V(G)\setminus\{v\}$, then in particular it cannot be
dominated by $3$ vertices from $D\setminus\{v\}$.  Hence
$D_1\subseteq \hat{D}$ and we can bound the size of $D_1$ by that of
$\hat{D}$.

\smallskip
\begin{lemma}\label{lem:D-hat}
  We have $D_1\subseteq \hat D$, $|\hat{D}\setminus D|< 4\gamma$, and
  $|\hat{D}|< 5\gamma$.
\end{lemma}
\begin{proof}
  \cref{lem:D-hat} follows the observation above together with
  \cref{lenzen-improved-planar}.
\end{proof}

From \cref{lem:D-hat} we can conclude that $|D_1|< 5\gamma$. However,
it is intuitively clear that every vertex that we pick from the
minimum dominating set~$D$ is optimal progress towards dominating the
whole graph. We will later show that this intuition is indeed true for
our algorithm, that is, our algorithm performs worst when
$D_1\cap D=\emptyset$, which will later in fact allow us to
estimate~$|D_1|<4\gamma$.

\smallskip

We mark the vertices of $D_1$ that we add to the dominating set in the
first phase of the algorithm as green, the neighbors of $D_1$ as
yellow and leave all other vertices red. Denote the set of red
vertices by~$R$, that is, $R=V(G)\setminus N[D_1]$.  For $v\in V(G)$
let $N_R(v)\coloneqq N(v)\cap R$ and $\delta_R(v)\coloneqq |N_R(v)|$
be the \emph{residual degree} of~$v$, that is, the number of neighbors
of $v$ that still need to be dominated.

\smallskip By definition of $D_1$, the neighborhood of every non-green
vertex can be dominated by at most $3$ other vertices. This holds true
as well for the subset $N_R(v)$ of neighbors that still need to be
dominated.  Let us fix such a small dominating set for the red
neighborhood of every non-green vertex.

\begin{tcolorbox}
  For every $v\in V(G)\setminus D_1$, we fix
  $A_v\subseteq V(G)\setminus \{v\}$ such that: \center
  $N_R(v)\subseteq N[A_v]$ and $|A_v|\leq 3$.
\end{tcolorbox}

There are potentially many such sets $A_v$ -- we fix one such set
arbitrarily.  Let us stress that even though we could compute the sets
$A_v$ in a local algorithm (making decisions based on vertex ids), we
only use these sets for our further argumentation and do not need to
compute them.

%% file: 4-high-degree.tex
% !TEX root = sirocco-main.tex

\section{Analyzing the local dominators}\label{sec:D2}

The second phase of our algorithm is inspired by results of Czygrinow
et al.~\cite{czygrinow2018distributed} and the greedy domination
algorithm for biclique-free graphs of~\cite{siebertz2019greedy}.
Czygrinow et al.~\cite{czygrinow2018distributed} defined the notion of
\emph{pseudo-covers}, which provide a tool to carry out a fine grained
analysis of vertices that can potentially belong to the sets $A_v$
used to dominate the red neighborhood~$N_R(v)$ of a vertex $v$. This
tool can in fact be applied to much more general graphs than planar
graphs, namely, to all graphs that exclude some complete bipartite
graph~$K_{t,t}$.  A refined analysis for classes of bounded expansion
was provided by Kublenz et al.~\cite{kublenz2020distributed}.  We
provide an even finer analysis for planar graphs on which we base a
second phase of our distributed algorithm.

We first describe what our algorithm computes, and then provide bounds
on the number of selected vertices. Intuitively, we select every pair
of vertices with sufficiently many neighbors in common.

\begin{tcolorbox}[colback=red!5!white,colframe=red!50!black]
\begin{itemize}
\item For $v\in V(G)$ let
  $B_v\coloneqq \{z\in V(G)\setminus \{v\}: |N_R(v)\cap N_R(z)|\geq
  10\}$.\smallskip
\item Let $W$ be the set of vertices $v\in V(G)$ such that
  $B_v \neq \emptyset$.\smallskip
\item Let $D_2\coloneqq \bigcup\limits_{v\in W} (\{v\}\cup B_v)$.
\end{itemize}
\end{tcolorbox}

Once $D_1$ has been computed in the previous phase, 2 more rounds of
communication are enough to compute the sets $B_v$ and $D_2$.
Before we update the residual degrees, let us analyze the sets $B_v$
and~$D_2$.  First note that the definition is symmetric: since
$N_R(v)\cap N_R(z)=N_R(z)\cap N_R(v)$ we have for all $v,z\in V(G)$ if
$z\in B_v$, then $v\in B_z$. In particular, if $v\in D_1$ or
$z\in D_1$, then $N_R(v)\cap N_R(z)=\emptyset$, which immediately
implies the following lemma.

\begin{lemma}\label{lem:WcapD1}
  We have $W\cap D_1=\emptyset$ and for every $v\in V(G)$ we have
  $B_v\cap D_1=\emptyset$.
\end{lemma}
%\begin{proof}
%  Let $v,z\in V(G)$. If $v\in D_1$, then
%  $N_R(v)=\emptyset$. Similarly if $z\in D_1$, then
%  $N_R(v)\cap N(z)=N(v)\cap N_R(z) =\emptyset$.
%\end{proof}

Now we prove that for every $v\in W$, the set $B_v$ cannot be too big,
and has nice properties.
\begin{lemma}\label{lem:dominating-dominators}
  For all vertices $v\in W$ we have

  \vspace{-5pt}
  \begin{itemize}
  \item $B_v \subseteq A_v$ (hence $|B_v|\le 3$), and \smallskip
  \item if $v\not\in \hat{D}$, then $B_v\subseteq D$.
  \end{itemize}
\end{lemma}

\begin{proof}
  Assume $A_v=\{v_1,v_2,v_3\}$ (a set of possibly not distinct
  vertices) and assume there exists
  $z\in V(G)\setminus \{v,v_1,v_2, v_3\}$ with
  $|N_R(v) \cap N_R(z)| \geq 10$.  As $v_1, v_2, v_3$ dominate $N_R(v)$,
  and hence also \mbox{$N_R(v)\cap N_R(z)$}, there must be some~$v_i$,
  $1\leq i\leq 3$, with
  \mbox{$|N_R(v) \cap N_R(z) \cap N[v_i]| \geq \lceil 10/3\rceil \geq
    4$}.  Therefore, \mbox{$|N_R(v) \cap N_R(z) \cap N(v_i)| \geq 3$},
  which shows that $K_{3,3}$ is a subgraph of~$G$, contradicting the
  assumption that $G$ is planar.

  If furthermore $v\not\in \hat{D}$, by definition of $\hat{D}$, we
  can find $w_1,w_2, w_3$ from $D$ that dominate $N(v)$, and in
  particular $N_R(v)$.  If $z\in V(G)\setminus \{v,w_1,w_2, w_3\}$
  with $|N_R(v) \cap N_R(z)| \geq 10$ we can argue as above to obtain
  a contradiction.
\end{proof}

% In the light of \cref{lem:dominating-dominators}, we select all
% paires of nodes with sufficiently large intersecting neighborhood.
%
% \begin{tcolorbox}
%   For $v\in V(G)$ let $B_v\coloneqq \{z\in V(G)\setminus \{v\}:
%   |N(v)\cap N(z)|\geq 19\}$.
% \end{tcolorbox}

% \begin{corollary}\label{cor:dominating-dominators}
%   For every vertex $v$, $B_v\subseteq A_v$, in particular,
%  $|B_v|\leq 6$ and if $v\not\in \hat{D}$, then $B_v\subseteq D$.
% \end{corollary}
%
% \begin{tcolorbox}
%   We define $W$ as the set of vertices $v\in V(G)$ such
%   that $B_v\neq \emptyset$. We define \[D_2\coloneqq \bigcup_{v\in W}
%   (\{v\}\cup B_v).\]
% \end{tcolorbox}
% \vspace{0mm}

% Our algorithm now proceeds as follows. Obviously, every vertex $v$
% can locally compute the set $B_v$. The algorithm adds the set $D_2$
% to the dominating set, removes~$D_2$ from the graph and marks all
% vertices dominated by $D_2$ as dominated.  \alex{here again 'remove'
% vertices?}

Let us now analyze the size of $D_2$. For this we refine the set $D_2$
and define
\begin{tcolorbox}
  \begin{enumerate}
    \item $D_2^1\coloneqq \bigcup_{v\in W\cap D}
    (\{v\}\cup B_v)$, \smallskip
    \item $D_2^2\coloneqq \bigcup_{v\in W\cap (\hat{D}\setminus D)}
    (\{v\}\cup B_v)$, and \smallskip
    \item $D_2^3\coloneqq \bigcup_{v\in W\setminus (D\cup \hat{D})}
    (\{v\}\cup B_v)$.
  \end{enumerate}
\end{tcolorbox}

\smallskip
Obviously $D_2=D_2^1\cup D_2^2\cup D_2^3$. We now bound the size of the
refined sets $D_2^1,D_2^2$ and $D_2^3$.

\begin{lemma}\label{lem:size-D21}
  $|D_2^1\setminus D|\leq 3\gamma$.
\end{lemma}
\begin{proof}
  We have
  \[|D_2^1\setminus D|= |\bigcup_{v\in W\cap D} (\{v\}\cup
    B_v)\setminus D|\leq |\bigcup_{v\in W\cap D}B_v|\leq \sum_{v\in
      W\cap D}|B_v|.\] By \cref{lem:dominating-dominators} we have
  $|B_v|\leq 3$ for all $v\in W$ and as we sum over $v\in W\cap D$ we
  conclude that the last term has order at most $3\gamma$.
\end{proof}

\begin{lemma}\label{lem:size-D22}
  $D_2^2 \subseteq \hat D$ and therefore
  $|D_2^2\setminus D|< 4\gamma$.
\end{lemma}
\begin{proof}
  Let $v\in \hat{D}\setminus D$ and let $z\in B_v$. By symmetry,
  $v\in B_z$ and according to \cref{lem:dominating-dominators}, if
  $z\not\in \hat{D}$, then $v\in D$.  Since this is not the case, we
  conclude that $z\in\hat{D}$.  Hence $B_v\subseteq \hat{D}$ and, more
  generally, $D_2^2\subseteq \hat{D}$.  Finally, according to
  \cref{lem:D-hat} we have $|\hat{D}\setminus D|<4\gamma$.
\end{proof}

Finally, the set $D_2^3$, which appears largest at first glance, was
actually already counted, as shown in the next lemma.
\begin{lemma}\label{lem:size-D23}
  $D_2^3\subseteq D_2^1$.
\end{lemma}
\begin{proof}
  If $v\not\in \hat{D}$, then $B_v\subseteq D$ by
  \cref{lem:dominating-dominators}.  Hence $v\in B_z$ for some
  $z\in D$, and $v\in D_2^1$.
\end{proof}

Again, it is intuitively clear that the situation when the sets
$D_2^i$ are large does not lead to the worst case for the overall
algorithm. For example, when $D_2^1$ is large we have added many
vertices of the optimum dominating set $D$. For a formal analysis, we
analyze the number of vertices of $D$ that have been selected so far.

\begin{tcolorbox}
  Let $\e\in [0,1]$ be such that $|(D_1 \cup D_2)\cap D| =\e\gamma$.
\end{tcolorbox}

% \begin{lemma}\label{lem:size-D2}
%   We have that $|D_2| < 3\gamma + 7\e\gamma$.
% \end{lemma}
% \begin{proof}
%   First, by \cref{lem:size-D22}, we have $|D_2^2|< 3\gamma+\e\gamma$. Then, with
%   \cref{lem:size-D21}, we have $|D_2^1|< 6\e\gamma$. Finally, with
%   \cref{lem:size-D23} we conclude that $|D_2|<3\gamma + 7\e\gamma$
% \end{proof}
%
% \begin{lemma}\label{lem:size-D1}
%   We have that $|D_1| < 3\gamma + \e\gamma$.
% \end{lemma}

% \alex{replacement for Lemmas 7 and 8 below}

\smallskip
\begin{lemma}\label{lem:size-D12}
  We have $|D_1\cup D_2| < 4\gamma+4\e\gamma$.
\end{lemma}
\begin{proof}
  By \cref{lem:size-D23} we have $D_2^3\subseteq D_2^1$, hence,
  $D_1\cup D_2=D_1\cup D_2^1\cup D_2^2$. By \cref{lem:D-hat} we have
  $D_1 \subseteq \hat D$ and by \cref{lem:size-D22} we also have
  $D_2^2 \subseteq \hat D$, hence $D_1\cup D_2^2\subseteq \hat D$.
  Again by \cref{lem:D-hat}, $|\hat D \setminus D|<4\gamma$ and
  therefore $|(D_1 \cup D_2^2 )\setminus D| < 4 \gamma$.

  We have $W\cap D\subseteq D_2^1\cap D$, hence with
  \cref{lem:dominating-dominators} we conclude that
  \[
    \big\vert D_2^1 \setminus D \big\vert \leq
    \Big\vert\bigcup\limits_{v\in D \cap D_2^1}B_v\Big\vert \leq
    \sum\limits_{v\in D \cap D_2^1} |B_v| \leq 3\e\gamma,
  \]
  hence $(D_1\cup D_2)\setminus D<4\gamma+3\epsilon\gamma$. Finally,
  $D_1\cup D_2=(D_1\cup D_2)\setminus D\cup ((D_1\cup D_2)\cap D)$ and
  with the definition of $\epsilon$ we conclude
  $|D_1\cup D_2|<4\gamma + 4\e\gamma$.
\end{proof}

The analysis of the next and final step of the algorithm will actually
show that the worst case is obtained when $\e=0$.

We now update the residual degrees, that is, we update $R$ as
$V(G)\setminus N[D_1\cup D_2]$ and for every vertex the number
$\delta_R(v)=N(v)\cap R$ accordingly.

%% file: 5-res-degree.tex
% !TEX root = sirocco-main.tex

\section{Greedy domination in planar graphs of maximum
residual degree}\label{sec:greedy}

We will show next that after the first two phases of the algorithm we
are in the very nice situation where all residual degrees are
small. This will allow us to proceed in a greedy manner.

  \begin{lemma}\label{lem:res-degree}
    For all $v\in V(G)$ we have $\delta_R(v)\leq 30$.
  \end{lemma}
  \begin{proof}
    First, every vertex of $D_1\cup D_2$ has residual degree $0$.
    Assume that there is a vertex $v$ of residual degree at least $31$.
    As $v$ is not in $D_1$, its $31$ non-dominated neighbors are
    dominated by a set $A_v$ of at most 3 vertices. Hence there is a
    vertex~$z$ (not in $D_1$ nor $D_2$) with $|N_R(v)\cap N_R[z]|\geq
    \lceil 31/3\rceil = 11$, hence, $|N_R(v)\cap N_R(z)|\geq 10$.
	This contradicts that $v$ is not in~$D_2$.
  \end{proof}

In the light of \cref{lem:res-degree}, we could now simply choose
$D_3$ as the set of elements not in $N[D_1\cup D_2]$.  We would get a
constant factor approximation, but not a very good one. Instead, we
now start to simulate the classical greedy algorithm, which in each
round selects a vertex of maximum residual degree. Here, we let all
non-dominated vertices that have a neighbor of maximum residual degree
choose such a neighbor as its dominator (or if they have maximum
residual degree themselves, they may choose themselves). In general
this is not possible for a LOCAL algorithm, however, as we established
a bound on the maximum degree we can proceed as follows.  We let
$i=30$. Every red vertex that has at least one neighbor of residual
degree $30$ arbitrarily picks one of them and elects it to the
dominating set. Then every vertex recomputes its residual degree and
$i$ is set to $29$. We continue until $i$ reaches $0$ when all
vertices are dominated. More formally, we define several sets as
follows.

\begin{tcolorbox}[colback=red!5!white,colframe=red!50!black]
  For $30\geq i\geq 0$,  for every $v\in R$ in parallel:\\[2mm]
  if there is some $u$ with $\dd_R(u)=i$ and ($\{u,v\}\in E(G)$ or $u=v$), then\\
  \mbox{ } $\dom_i(v)\coloneqq \{u\}$ (pick one such $u$ arbitrarily),\\
  \mbox{ } $\dom_i(v)\coloneqq \emptyset$ otherwise.
  \begin{itemize}
    \item $R_i \coloneqq R$ \hfill \textit{\small What currently remains to be dominated}
    \item $\Delta_i \coloneqq \bigcup\limits_{v\in R} \dom_i(v)$ \hfill \textit{\small What we pick in this step}
    \item $R \coloneqq R \setminus N[\Delta_{i}]$ \hfill \textit{\small Update red vertices}
  \end{itemize}
  Finally, $D_3:=  \bigcup\limits_{1\le i\le 30} \Delta_i$.
\end{tcolorbox}

\smallskip
Let us first prove that the algorithm in fact computes a dominating set.
\begin{lemma}\label{lem:correctness}
  When the algorithm has finished the iteration with parameter
  $i\geq 1$, then all vertices have residual degree at most $i-1$.
\end{lemma}

In particular, after finishing the iteration with parameter $1$, there
is no vertex with residual degree $1$ left and in the final round all
non-dominated vertices choose themselves into the dominating
set. Hence, the algorithm computes a dominating set of $G$.

\begin{proof}
  By induction, before the iteration with parameter $i$, all vertices
  have residual at most $i$. Assume $v$ has residual degree $i$ before
  the iteration with parameter $i$.  In that iteration, all
  non-dominated neighbors of $v$ choose a dominator (possibly $v$, then
  the statement is trivial),
  hence, are removed from $R$. It follows that the residual degree of $v$ after
  the iteration is $0$. Hence, after this iteration and before the
  iteration with parameter $i-1$, we are left with vertices of
  residual degree at most $i-1$.
\end{proof}

We now analyze the sizes of the sets $\Delta_i$ and $R_i$. The first
lemma follows from the fact that every vertex chooses at most one
dominator.

\begin{lemma}\label{lem:total-h}
  For every $i\le 30$, $\sum\limits_{j\le i}|\Delta_j| \le |R_i|$.
\end{lemma}
\begin{proof}
  The vertices of $R_i$ are those that remain to be dominated in the
  last $i$ rounds of the algorithm. As every vertex that remains to be
  dominated chooses at most one dominator in one of the rounds
  $j\leq i$, the statement follows.
\end{proof}

\pagebreak
As the vertices of $D$ that still dominate non-dominated vertices also
have bounded residual degree, we can conclude that not too many
vertices remain to be dominated.
\begin{lemma}\label{lem:h1}
  For every $i\le 30$, $|R_i| \le (i+1)(1-\e)\gamma$.
\end{lemma}
\begin{proof}
  First note that for every $i$,
  $D\setminus (D_1\cup D_2\cup \bigcup_{j>i}\Delta_j)$ is a dominating
  set for $R_i$; additionally each vertex in this set has residual
  degree at most $i$.  And finally, this set is a subset of
  $D\setminus (D_1\cup D_2)$. Hence by the definition of $\e$, we get
  that
  $|D\setminus (D_1\cup D_2\cup \bigcup_{j>i}\Delta_j)|\le
  (1-\e)\gamma$. As every vertex dominates its residual neighbors and
  itself, we conclude $|R_i|\le (i+1)(1-\e)\gamma$.
\end{proof}

The next lemma shows that we cannot pick too many vertices of high
residual degree. This follows from the fact that planar graphs have
bounded edge density.

\begin{lemma}\label{lem:delta}
  For every $7\le i\le 30$, $|\Delta_i| \le \frac{3|R_i|}{i-6}$.
\end{lemma}
\begin{proof}
  Let $7\le i\le 30$ be an integer. We bound the size of $\Delta_i$
  by a counting argument, using that $G$ (as well as each of its
  subgraphs) is planar, and can therefore not have to many edges.

  Let $J := G[\Delta_i]$ be the subgraph of $G$ induced by the
  vertices of $\Delta_i$, which all have residual degree~$i$. Let
  $K := G[\Delta_i \cup (N[\Delta_i]\cap R_i)]$ be the subgraph of $G$
  induced by the vertices of $\Delta_i$ together with the red
  neighbors that these vertices dominate.

  As $J$ is planar, $|E(J)| < 3|V(J)| = 3|\Delta_i|$. As every
  vertex of $J$ has residual degree exactly $i$, we get
  $|E(K)| \geq i\Delta_i - |E(J)| > (i-3)|\Delta_i|$ (we have to
  subtract $|E(J)|$ to not count twice the edges of $K$ that are
  between two vertices of $J$).  We also have that
  $|V(K)| \le |V(J)| + |R_i|$. We finally apply Euler's formula again
  to $K$ and get that $|E_K| < 3|V_K|$ hence
  $(i-3)|\Delta_i| < 3|\Delta_i| + 3|R_i|$. Therefore
  $|\Delta_i|< \frac{3|R_i|}{i-6}$.
\end{proof}

Finally, we can give a lower bound on how many elements are newly
dominated by the chosen elements of high residual degree.

\begin{lemma}\label{lem:h2}
  For every $1\leq i\leq 29$, $|R_i| \le |R_{i+1}| - \frac{(i-5)|\Delta_{i+1}|}{3}$.
\end{lemma}

% Intuitively, these lemmas follows simple reasoning.
% \cref{lem:total-h} holds because we make sure to never take more
% than what there is left to dominate. Then \cref{lem:h1} holds
% because the elements of $D$ that have not been selected in $D_1$ nor
% $D_2$ form a dominating set for every $R_i$. Additionally, these
% dominators have bounded residual degree.  \cref{lem:delta} holds
% thank to an argument similar to the one of \cref{lem:lenzen}: the
% edge density of the graph induced by $R_i$ and $\Delta_i$ is at most
% 3, and every vertices of $\Delta_i$ has high degree.  And finally,
% \cref{lem:h2} is prooved with a similar argument: once $\Delta_i$ is
% fixed, its neighborhood cannot be too small (this would create a
% dense subgraph). Therefore we can give a lower bound to the number
% of elements that are newly dominated.

\begin{proof}
  Similarly to the proof of \cref{lem:delta} (by replacing $i$ by
  $i+1$), we define $J := G[\Delta_{i+1}]$ and
  $K:= G[\Delta_{i+1} \cup (N[\Delta_{i+1}]\cap R_{i+1})]$.

  We then
  replace the bound $|V(K)| \le |V(J)| + |R_{i+1}|$ by
  $|V(K)| \le |V(J)| + |N[\Delta_{i+1}]\cap R_{i+1}|$.

  We then get:
  \[|E_K| \leq 3 |V_K|, \]
  \[(i+1)|\Delta_{i+1}| - 3|\Delta_{i+1}| \leq 3(|\Delta_{i+1}| +
    |N[\Delta_{i+1}]\cap R_{i+1}|)\text{, and}\]
  \[ |N[\Delta_{i+1}]\cap R_{i+1}| \geq \frac{(i+1-6)|\Delta_{i+1}|}{3}.\]

  Now, as $R_i = R_{i+1} \setminus N[\Delta_{i+1}]$, we have
  $|R_i| \le |R_{i+1}| - |N[\Delta_{i+1}]\cap R_{i+1}|\leq |R_{i+1}| -
  \frac{(i+1-6)|\Delta_{i+1}|}{3}$.
\end{proof}

We now formulate (and present in \cref{sec:linear-prog}) a linear program to
maximize $|D_3|$ under these constraints. As a result we conclude the
following lemma.

\begin{lemma}\label{lem:size-D3}
  $|D_3|\le 15.9(1-\e)\gamma$.
\end{lemma}

%%%%%%%%%%%%%%%%%%%%%%%%%%%%%%%%%%%%%%%%%%%%%%%%%%%%%%%%%%%%%%%%%%%%%%%%%%%%%%%%
%%%%%%%%%%%%%%%%%%%%%%%%%%%%%%%%%%%%%%%%%%%%%%%%%%%%%%%%%%%%%%%%%%%%%%%%%%%%%%%%
\section{Summarizing the planar case}
%%%%%%%%%%%%%%%%%%%%%%%%%%%%%%%%%%%%%%%%%%%%%%%%%%%%%%%%%%%%%%%%%%%%%%%%%%%%%%%%
%%%%%%%%%%%%%%%%%%%%%%%%%%%%%%%%%%%%%%%%%%%%%%%%%%%%%%%%%%%%%%%%%%%%%%%%%%%%%%%%

We already noted that the definition of $D_3$ implies that
$D_1\cup D_2\cup D_3$ is a dominating set of $G$. We now conclude the
analysis of the size of this computed set.  First, by
\cref{lem:size-D12} we have $|D_1\cup D_2|<4\gamma+4\e\gamma$.  Then,
by \cref{lem:size-D3} we have $|D_3|\le 15.9(1-\e)\gamma$.
Therefore $|D_1 \cup D_2 \cup D_3| < 19.9\gamma -11.9\e\gamma$.
As $\e\in[0,1]$, this is maximized when $\e=0$. Hence
\mbox{$|D_1 \cup D_2 \cup D_3|< 19.9 \gamma$}.

\begin{theorem}\label{thm:planar}
  There exists a distributed LOCAL algorithm that, for every planar
  graph $G$, computes in a constant number of rounds a dominating set
  of size at most $20\gamma(G)$.
\end{theorem}

%% file: 6-girth.tex
% !TEX root = sirocco-main.tex

\section{Restricted classes of planar graphs}
In this section we further restrict the input graphs, requiring e.g.\ planarity
together with a lower bound on the girth. Our algorithm works exactly as before, however,
using different parameters. From the different edge densities and
chromatic numbers of
the restricted graphs we will then derive different constants and
as a result a better approximation factor. Throughout this section
we use the same notation as in the first part of the paper and state
in the adapted lemmas with the same numbers as in the first part
the adapted sizes of the respective sets.

As in the general case in the first phase we begin by computing
the set $D_1$ and analyzing it in terms of the auxiliary set $\hat{D}$.

%The proofs of most lemmas
%can be easily adapted from the general case and we will only go into
%detail when needed. In \cref{sec:triangle-free} we prove that on
%triangle-free planar graphs we can compute a 16-approximation, improving
%the currently best known bound of 32. In \cref{sec:girth} we prove
%that on planar graphs of girth 5 we can compute a 9-approximation,
%improving the best known bound of 18, and finally in \cref{sec:outer}
%we prove that on outerplanar graphs our algorithm computes a 13-approximation. This does not improve the tight known bound
%of 5, however, it demonstrates that our algorithm works robustly
%on this interesting subclass.

\setcounter{lemma}{2}

\begin{acor}\label{a-lenzen-improved-planar}\hspace{-1.7mm}\textbf{.}

  \vspace{-6mm}
  \textit{\begin{enumerate}
    \item If $G$ is bipartite, then $|\hat{D}\setminus D| < 2\gamma$.\smallskip
    \item If $G$ is triangle-free, outerplanar, or has girth 5,
      then $|\hat{D}\setminus D| < 3\gamma$.
  \end{enumerate}}
\end{acor}
\begin{proof}
This is immediate from \cref{lem:densities} and \cref{lenzen-improved}.
\end{proof}

The inclusion $D_1\subseteq \hat D$ continues to hold and the bound
on the sizes as stated in the next lemma is again a direct consequence of the corollary.

\smallskip
\begin{alemma}\label{alem:D-hat}\hspace{-1.7mm}\textbf{.}
    We have $D_1\subseteq \hat D$, and
  \textit{\begin{enumerate}
    \item if $G$ is bipartite,
      then $|\hat{D}\setminus D| < 2\gamma$ and $|\hat{D}|< 3\gamma$.\smallskip
    \item if $G$ is triangle-free, outerplanar, or has girth 5,
      then $|\hat{D}\setminus D| < 3\gamma$ and $|\hat{D}|< 4\gamma$.
  \end{enumerate}}
\end{alemma}

In case of triangle-free planar graphs (in particular in the case of bipartite
planar graphs) we proceed with the second phase exactly as in the second phase of
the general algorithm (\cref{sec:D2}), however, the parameter 10 is replaced by
the parameter $7$. In  case of planar graphs of girth at least five or outerplanar
graphs, we simply set $D_2=\emptyset$.

\begin{tcolorbox}[colback=red!5!white,colframe=red!50!black]
  If $G$ is triangle-free:

  \begin{itemize}
    \item For $v\in V(G)$ let $B_v\coloneqq \{z\in V(G)\setminus
      \{v\}: |N_R(v)\cap N_R(z)|\geq 7\}$.\smallskip
    \item Let $W$ be the set of vertices $v\in V(G)$ such
      that $B_v \neq \emptyset$.\smallskip
    \item Let $D_2\coloneqq \bigcup\limits_{v\in W} (\{v\}\cup B_v)$.
  \end{itemize}

  If $G$ has girth at least $5$ or $G$ is outerplanar, let $D_2=\emptyset$.
\end{tcolorbox}

\cref{lem:WcapD1} is based only on the definition of $B_v$ and $W$ and
does not use particular properties of planar graphs, hence, it also holds
in the restricted case and we recall it for convenience.

\begin{lemma}\label{alem:WcapD1}
  We have $W\cap D_1=\emptyset$ and for every $v\in V(G)$ we have
  $B_v\cap D_1=\emptyset$.
\end{lemma}

The next lemma uses the triangle-free property.

\begin{alemma}\label{alem:dominating-dominators}\hspace{-1.7mm}\textbf{.}
  If $G$ is triangle-free, then for all vertices $v\in W$ we have

  \vspace{-5pt}
  \begin{itemize}
  \item $B_v \subseteq A_v$ (hence $|B_v|\le 3$), and \smallskip
  \item if $v\not\in \hat{D}$, then $B_v\subseteq D$.
  \end{itemize}
\end{alemma}
\begin{proof}
  Assume $A_v=\{v_1,v_2, v_3\}$ and assume there is $z\in V(G)\setminus \{v,v_1,v_2, v_3\}$
  with $|N_R(v) \cap N_R(z)| \geq 7$.
  As the vertices $v_1, v_2, v_3$ dominate $N_R(v)$, and hence $N_R(v)\cap N_R(z)$,
  there must be some~$v_i$, $1\leq i\leq 3$, with
  \mbox{$|N_R(v) \cap N_R(z) \cap N[v_i]| \geq \lceil 7/3\rceil \geq 3$}.
  Then on of the following holds: either
  \mbox{$|N_R(v) \cap N_R(z) \cap N(v_i)| \geq 3$},  or
  \mbox{$|N_R(v) \cap N_R(z) \cap N(v_i)| =2$}.
  The first case shows that $K_{3,3}$ is a subgraph of~$G$
  contradicting the assumption that $G$ is planar.
  The second case implies that $v_i\in N_R(v)$. In this situation, by picking
  $w \in N_R(v) \cap N_R(z) \cap N(v_i)$, we get that $(v,v_i,w)$ is a triangle,
  hence we also reach a contradiction.

  If furthermore $v\not\in \hat{D}$, by definition of $\hat{D}$,
  we can find $w_1,w_2, w_3$ from $D$
  that dominate $N(v)$, and in particular $N_R(v)$.
  If $z\in V(G)\setminus \{v,w_1,w_2, w_3\}$
  with $|N_R(v) \cap N_R(z)| \geq 7$ we can argue as above to obtain
  a contradiction.
\end{proof}

For our analysis we again split $D_2$ into three sets $D_2^1, D_2^2$ and
$D_2^3$. The next lemmas hold also for the restricted cases. We repeat
them for convenience.

\begin{alemma}\label{alem:size-D21}
  If $G$ is triangle-free, then $|D_2^1\setminus D|\leq 3\gamma$.
\end{alemma}

\begin{alemma}\label{alem:size-D22}
  If $G$ is triangle-free, then $D_2^2 \subseteq \hat D$ and therefore
  $|D_2^2\setminus D|< 3\gamma$.
\end{alemma}

\begin{alemma}\label{alem:size-D23}
  If $G$ is triangle-free, then $D_2^3\subseteq D_2^1$.
\end{alemma}

Again, for a fine analysis, we
analyze the number of vertices of $D$ that have been selected
so far and let $\e\in [0,1]$ be such that $|(D_1 \cup D_2)\cap D| =\e\gamma$.

\begin{alemma}\label{alem:size-D12}
\mbox{ }
\vspace{-1mm}
\begin{enumerate}
\item If $G$ is bipartite, then $|D_1\cup D_2| < 2\gamma+4\e\gamma$.\smallskip
\item If $G$ is triangle-free, then $|D_1\cup D_2| < 3\gamma+4\e\gamma$.\smallskip
\item If $G$ has girth at least $5$ or is outerplanar, then $|D_1\cup D_2| < 3\gamma+\e\gamma$.
\end{enumerate}
\end{alemma}
\begin{proof}
If $G$ is outerplanar or $G$ has girth at least $5$, then $D_2=\emptyset$.
By \cref{alem:D-hat} we have $D_1\subseteq \hat{D}$ and
$|\hat{D}\setminus D|<3\gamma$, hence $(D_1\cup D_2)\setminus D<3\gamma$.

If $G$ is triangle-free, by \cref{alem:size-D23} we have $D_2^3\subseteq D_2^1$, hence,
  $D_1\cup D_2=D_1\cup D_2^1\cup D_2^2$. By \cref{alem:D-hat} we have
  $D_1 \subseteq \hat D$ and by \cref{alem:size-D22} we also have
  $D_2^2 \subseteq \hat D$, hence $D_1\cup D_2^2\subseteq \hat D$.
  Again by \cref{alem:D-hat}, if $G$ is bipartite, then
  $|\hat D \setminus D|<2\gamma$, therefore $|(D_1 \cup D_2^2 )\setminus D| < 2\gamma$, and if $G$ is triangle-free,
  then $|\hat D \setminus D|<3\gamma$,
  therefore $|(D_1 \cup D_2^2 )\setminus D| < 3\gamma$.
  We have $W\cap D\subseteq D_2^1\cap D$, hence with
  \cref{alem:dominating-dominators} we conclude that
  \[
    \big\vert D_2^1 \setminus D \big\vert \leq
    \Big\vert\bigcup\limits_{v\in D \cap D_2^1}B_v\Big\vert \leq
    \sum\limits_{v\in D \cap D_2^1} |B_v| \leq 3\e\gamma,
  \]
  hence $(D_1\cup D_2)\setminus D<2\gamma+3\epsilon\gamma$ if
  $G$ is bipartite and $(D_1\cup D_2)\setminus D<3\gamma+3\epsilon\gamma$
  if $G$ is triangle-free.

  Finally,
  $D_1\cup D_2=(D_1\cup D_2)\setminus D\cup (D_1\cup D_2)\cap D$ and
  with the definition of $\epsilon$ we conclude

  \vspace{-2mm}
  \begin{enumerate}
  \item $|D_1\cup D_2|<2\gamma + 4\e\gamma$ if $G$ is bipartite, \smallskip
  \item $|D_1\cup D_2|<3\gamma + 4\e\gamma$ if $G$ is triangle-free, \smallskip
  \item $|D_1\cup D_2|<3\gamma + \e\gamma$ if $G$ has girth at least
  $5$ or is outerplanar.
  \end{enumerate}

\end{proof}

Again, we now update the residual degrees, that is, we update
$R$ as $V(G)\setminus N[D_1\cup D_2]$ and for every vertex the
number $\delta_R(v)=N(v)\cap R$ accordingly and proceed with
the third phase.

\begin{alemma}\label{alem:res-degree}\hspace{-1.7mm}\textbf{.}

\vspace{-6mm}
\textit{\begin{enumerate}
\item If $G$ is triangle-free, then for all $v\in V(G)$ we have $\delta_R(v)\leq 18$.\smallskip
\item If $G$ has girth at least $5$, then for all $v\in V(G)$ we have $\delta_R(v)\leq 3$.\smallskip
\item If $G$ is outerplanar, then for all $v\in V(G)$ we have $\delta_R(v)\leq 9$.
\end{enumerate}}
\end{alemma}
\begin{proof}
  Every vertex of $D_1\cup D_2$ has residual degree $0$, hence, we
  need to consider only vertices that are not in $D_1$ or $D_2$.

  First assume that the graph is triangle-free and
  assume that there is a vertex $v$ of residual degree at least $19$.
  As $v$ is not in $D_1$, its $19$ non-dominated
  neighbors are dominated by a
  set $A_v$ of at most~3 vertices. Hence, there is vertex $z$ (not in $D_1$
  nor $D_2$) dominating at least $\lceil 19/3\rceil = 7$ of them.
  Here, $z$ cannot be one of these 7 vertices, otherwise it would be connected
  to $v$ and there would be a triangle in the graph.
  Therefore we
  have $|N_R(v)\cap N_R(z)|\geq 7$, contradicting that $v$ is not in~$D_2$.

  Now assume that $G$ has girth at least $5$ and
  assume that there is a vertex $v$ of residual degree at least~$4$.
  As $v$ is not in $D_1$, its $4$ non-dominated
  neighbors are dominated by a
  set $A_v$ of at most~3 vertices. Hence, there is vertex $z$ (not in $D_1$
  nor $D_2$) dominating at least $\lceil 4/3\rceil = 2$ of them.
  Here, $z$ cannot be one of these 2 vertices, otherwise it would be connected
  to $v$ and there would be a triangle in the graph. However, $z$ can
  also not be any other vertex, as otherwise we find a cycle of length $4$,
  contradicting that~$G$ has girth at least $5$.

  Finally, assume that $G$ is outerplanar and
  assume that there is a vertex $v$ of residual degree at least $10$.
  As $v$ is not in $D_1$, its $10$ non-dominated
  neighbors are dominated by a
  set $A_v$ of at most~3 vertices. Hence, there is vertex $z$ (not in $D_1$
  nor $D_2$) dominating at least $\lceil 10/3\rceil = 4$ of them.
  Therefore $|N(v)\cap N(z)|\geq 3$, and we find a $K_{2,3}$ as a
  subgraph, contradicting that $G$ is outerplanar.
\end{proof}

We proceed to compute a dominating set of the remaining vertices
as before for the respective number of rounds.

\setcounter{lemma}{12} % this is to match the non-adapted lemmas

\begin{alemma}\label{alem:total-H}
  If $G$ is triangle-free or outerplanar,
  for every $1\le i$, $\sum\limits_{j\le i}|\Delta_j| \le |R_i|$.
\end{alemma}

\begin{alemma}\label{alem:tri-h1}
  If $G$ is triangle-free or outerplanar,
  for every $1\le i$, $|R_i| \le (i+1)(1-\e)\gamma$.
\end{alemma}

\begin{alemma}\label{alem:tri-delta}
  If $G$ is triangle-free or outerplanar,
  for every $5\le i$, $|\Delta_i| \le \frac{2|R_i|}{i-4}$.
\end{alemma}

\begin{alemma}\label{alem:tri-h2}
  If $G$ is triangle-free or outerplanar,
  for every $1\le i$, $|R_i| \le |R_{i+1}| - \frac{(i-3)|\Delta_{i+1}|}{2}$.
\end{alemma}

The proofs of \cref{alem:total-H} to \ref{alem:tri-h2} are copies of the ones
for \cref{lem:total-h,lem:h1,lem:delta,lem:h2}, with the execption that the edge
density of $3$ for planar graphs if now replaced by $2$ for triangle-free and
outerplanar.
Similarly to \cref{lem:size-D3} we formulate (and present in
\cref{sec:linear-prog}) a linear program to maximize $|D_3|$ under these
constraints, yielding the following lemma.

\begin{alemma}\label{alem:size-D3}\hspace{-1.7mm}\textbf{.}

\vspace{-6mm}
\textit{\begin{enumerate}
\item If $G$ is triangle-free, then  $|D_3|\le 10.5(1-\e)\gamma$.\smallskip
\item If $G$ has girth at least $5$, then $|D_3|\le 4(1-\e)\gamma$.\smallskip
\item If $G$ is outerplanar, then $|D_3|\le 8.6(1-\e)\gamma$.
\end{enumerate}}
\end{alemma}

\begin{theorem}\label{thm:tri}
There exists a distributed LOCAL algorithm that, for every triangle free planar
graph $G$, computes in a constant number of rounds a dominating set
of size at most $14\gamma(G)$.
\end{theorem}
\begin{proof}
By
\cref{alem:size-D12} we have $|D_1\cup D_2|<3\gamma+4\e\gamma$.  Then,
by \cref{alem:size-D3} we have $|D_3|\le 10.5(1-\e)\gamma$.
Therefore $|D_1 \cup D_2 \cup D_3| < 13.5\gamma -6.5\e\gamma$.
As $\e\in[0,1]$, this is maximized when $\e=0$. Hence
\mbox{$|D_1 \cup D_2 \cup D_3|< 13.5 \gamma$}.
\end{proof}

\begin{theorem}\label{thm:bip}
  There exists a distributed LOCAL algorithm that, for every bipartite planar graph
  $G$, computes in a constant number of rounds a
  dominating set of size at most $13\gamma(G)$.
\end{theorem}
\begin{proof}
By
\cref{alem:size-D12} we have $|D_1\cup D_2|<2\gamma+4\e\gamma$.  Then,
by \cref{alem:size-D3} we have $|D_3|\le 10.5(1-\e)\gamma$.
Therefore $|D_1 \cup D_2 \cup D_3| < 12.5\gamma -6.5\e\gamma$.
As $\e\in[0,1]$, this is maximized when $\e=0$. Hence
\mbox{$|D_1 \cup D_2 \cup D_3|< 12.5 \gamma$}.
\end{proof}

\begin{theorem}\label{thm:girth}
  There exists a distributed LOCAL algorithm that, for every planar graph
  $G$ of girth at least~$5$, computes in a constant number of rounds a
  dominating set of size at most $7\gamma(G)$.
\end{theorem}
\begin{proof}
By
\cref{alem:size-D12} we have $|D_1\cup D_2|<3\gamma+\e\gamma$.  Then,
by \cref{alem:size-D3} we have $|D_3|\le 4(1-\e)\gamma$.
Therefore $|D_1 \cup D_2 \cup D_3| < 7\gamma -3\e\gamma$.
As $\e\in[0,1]$, this is maximized when $\e=0$. Hence
\mbox{$|D_1 \cup D_2 \cup D_3|< 7 \gamma$}.
\end{proof}

\begin{theorem}\label{thm:outer}
  There exists a distributed LOCAL algorithm that, for every outerplanar graph
  $G$, computes in a constant number of rounds a
  dominating set of size at most $12\gamma(G)$.
\end{theorem}
\begin{proof}
By
\cref{alem:size-D12} we have $|D_1\cup D_2|<3\gamma+\e\gamma$.  Then,
by \cref{alem:size-D3} we have $|D_3|\le 8.6(1-\e)\gamma$.
Therefore $|D_1 \cup D_2 \cup D_3| < 11.6\gamma -7.6\e\gamma$.
As $\e\in[0,1]$, this is maximized when $\e=0$. Hence
\mbox{$|D_1 \cup D_2 \cup D_3|< 11.6 \gamma$}.
\end{proof}

%% file: 8-conclusion.tex
\section{Conclusion}

  We provided a new LOCAL algorithm that computes a
  20\hspace{1pt}-\hspace{1pt}approximation of
  a minimum dominating set in a planar graph in a constant number of
  rounds. Started with different parameters, the algorithm works also
  for several restricted cases of planar graphs. We showed that
  it computes a 14\hspace{1pt}-\hspace{1pt}approximation for
  triangle-free planar graphs, a 13\hspace{1pt}-\hspace{1pt}approximation
  for bipartite planar graphs, a 7\hspace{1pt}-\hspace{1pt}approximation
  for planar graphs of girth 5 and a 12\hspace{1pt}-\hspace{1pt}approximation
  for outerplanar graphs. In all cases except for the outerplanar case,
  where an optimal bound of 5 was already known, our algorithm
  improves on the previously best known approximation factors.
  This improvement is most significant in the case of general planar
  graphs, where the previously best known factor was 52.
  While we could tighten the gap between the best known lower bound of 7
  and upper bound of 52, there is still some room for improvement. We
  believe that the optimum approximation rate is much closer to $7$ than
  to $20$.

%% file: 7-cplex.tex
% !TEX root = sirocco-main.tex

\section{The linear program}\label{sec:linear-prog}

In this final section we present our formulation of the constraints as a
linear program as well as the resulting bounds on how many vertices of
the specific residual degrees can be found in the worst case.
We formulate the constraints of \cref{lem:total-h,lem:h1,lem:delta,lem:h2} in a
straight forward way and remove the
$(1-\e)\gamma$ factor, which is then added to the result.
This reasoning is correct thanks to the fact that all constraints are linear
equations; we formally prove it below.

\smallskip
Define $r_i := \frac{|R_i|}{(1-\e)\gamma}$ and  $d_i:=\frac{|\Delta_i|}{(1-\e)\gamma}$.
Then the constraints of \cref{lem:total-h,lem:h1,lem:delta,lem:h2} imply respectively:
\begin{itemize}
  \item For every $0\le i\le 30$:\quad $r_i \ge \sum\limits_{j\le i} d_j$.\\ \smallskip
  \item For every $0\le i\le 30$:\quad $r_i \le i+1$.\\ \smallskip
  \item For every $7\le i\le 30$:\quad $d_i \le \frac{3 r_i}{i-6}$.\\ \smallskip
  \item For every $0\le i\le 29$:\quad $r_i \le r_{i+1} - \frac{(i-5)d_{i+1}}{3}$.
\end{itemize}

We then run the linear program with these variables; finally we provide the bound for $D_3$ using:
\[|D_3| ~=~ \sum\limits_{i\le 30} |\Delta_i|
~=~ \sum\limits_{i\le 30} d_i(1-\e)\gamma
~=~ (1-\e)\gamma\sum\limits_{i\le 30}d_i.\]

Before showing the code and the results, we briefly explain what we expect as a
result for these linear programs.

\subsection{Interpretation of the results}
In all four cases, our sets of equations yield similar looking results.
The step $3$ can roughly be decomposed into two.

First, for several values of
$i$, we have very small $d[i]$. We exactly have $d[i]$ such that
given $r[i]=i+1$ we get $r[i-1]=i$. Intuitively, picking less element in $d[i]$
is not the worst case as $r[i-1]$ cannot be bigger than $i$ by \cref{lem:h1}.
So it is ``free'' to take at least that many vertices.
It is also not the worst case if more elements are picked, because then $r[i]$
would shrink drastically, making the forthcoming $d[j]_{j<i}$ much smaller.

Second, there is a turning point. It occurs a little bit above the average
degree of planar graphs; so the number of vertices of degree 9 for example is
not so small. This is when \cref{lem:total-h} become predominant:
``Overall, we do not take more dominators that there are vertices to dominate.''
So in one round every vertex gets picked and the algorithm stops.
This turning point is $i=9$ for planar graphs.

We did not manage to formally prove this statement, but it was confirmed for
these cases by the linear programs.

\newpage
\subsection{The linear program for planar graphs}

\begin{verbatim}
//the ranges i can have
range I  = 1..30;
range I2 = 1..29;
range I3 = 7..30;

//decision variables as arrays
dvar float+ d[I];
dvar float+ r[I];

//maximize the sum of A_i
maximize sum(i in I) (d[i]);

// our equations
subject to
{
  // By lemma 13
  forall(i in I) r[i] >= sum(x in 1..i)d[x];

  // By lemma 14
  forall(i in I) r[i] <= i+1;

  // By lemma 15
  forall(i in I3) d[i] <= (3 * r[i])  /  ( i-6 );

  // By Lemma 16
  forall(i in I2) r[i] <=  r[i+1]  -  ((( i-5 ) * d[i+1]) / 3);
}
\end{verbatim}

% \subsection{The degree distribution in general planar graphs}

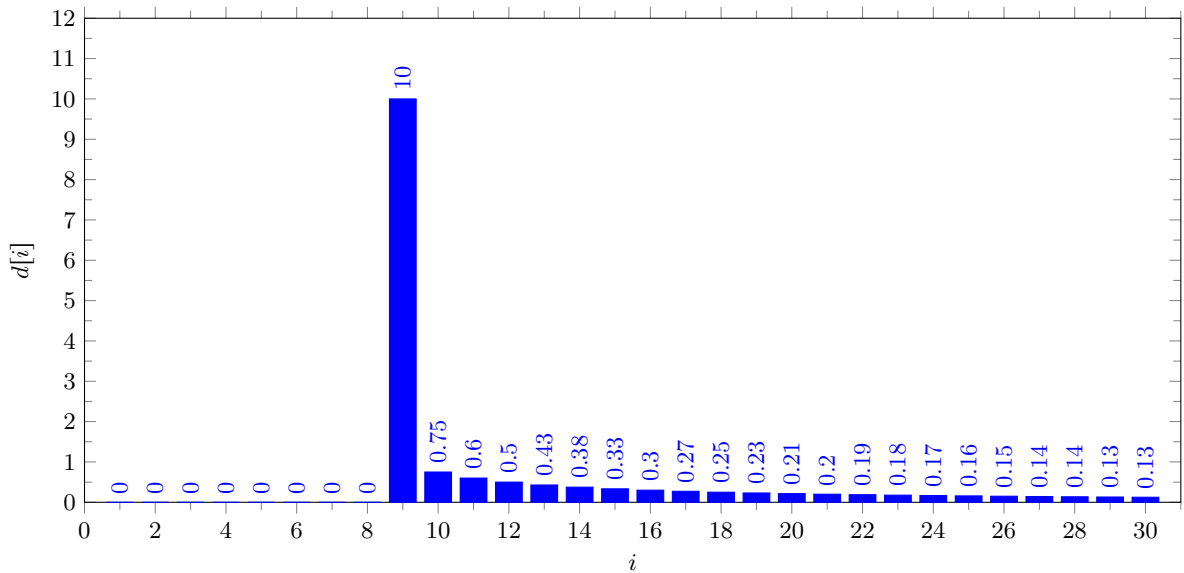
\begin{figure}
  \begin{tikzpicture}
    \begin{axis}[
      ybar,
      xmin = 0, xmax = 31,
      ymin = 0, ymax = 12,
      xtick distance = 2,
      ytick distance = 1,
      minor tick num = 1,
      width = \textwidth,
      height = \textwidth*0.5,
      xlabel = {$i$},
      ylabel = {$d[i]$},]
    ]

    %plot the first function
    \addplot +[
      ybar,
      fill=blue,
      nodes near coords,
      nodes near coords style = {anchor=west, rotate=90}
    ] file[skip first] {results_normal.txt};

    \end{axis}
  \end{tikzpicture}
  \caption{The degree distribution in general planar graphs}
\end{figure}

\pagebreak
\subsection{The linear program for triangle-free planar graphs}

\begin{verbatim}
//Tri-Free
  //the ranges i can have
  range I  = 1..18;
  range I2 = 1..17;
  range I3 = 5..18;

  //decision variables as arrays
  dvar float+ d[I];
  dvar float+ r[I];

  //maximize the sum of A_i
  maximize sum(i in I) (d[i]);

  // our equations
  subject to
  {
    // By Adapted lemma 13
    forall(i in I) r[i] >= sum(x in 1..i)d[x];

    // By Adapted lemma 14
    forall(i in I) r[i] <= i+1;

    // By Adapted lemma 15
    forall(i in I3) d[i] <= (2 * r[i])  /  ( i-4 );

    // By Adapted lemma 16
    forall(i in I2) r[i] <=  r[i+1]  -  ((( i-3 ) * d[i+1]) / 2);
  }
\end{verbatim}

\begin{figure}
  \begin{tikzpicture}
    \begin{axis}[
      ybar,
      xmin = 0, xmax = 20,
      ymin = 0, ymax = 9,
      xtick distance = 5,
      ytick distance = 1,
      minor tick num = 1,
      width = \textwidth,
      height = \textwidth*0.5,
      xlabel = {$i$},
      ylabel = {$d[i]$},]
    ]

    %plot the first function
    \addplot +[
      ybar,
      fill=blue,
      nodes near coords,
      nodes near coords style = {anchor=west, rotate=90}
    ] file[skip first] {results_tri-free.txt};

    \end{axis}
      \end{tikzpicture}
  \caption{The degree distribution in triangle-free planar graphs}
\end{figure}
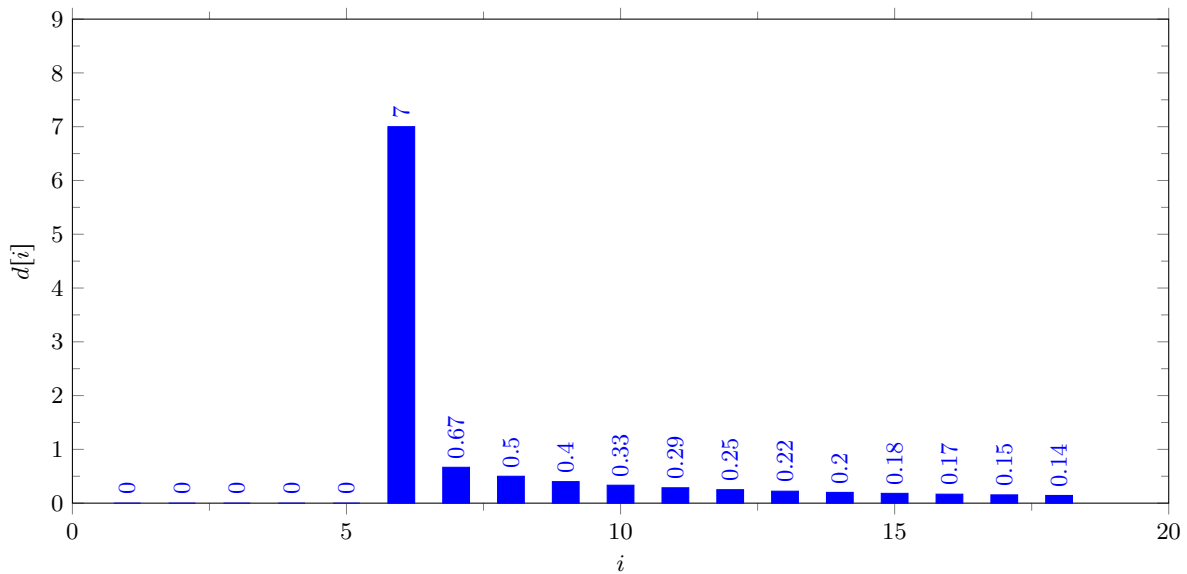

\pagebreak
\subsection{The linear program for outerplanar graphs}

\begin{verbatim}
//outerplanar

  //the ranges i can have
  range I  = 1..9;
  range I2 = 1..8;
  range I3 = 5..9;

  //decision variables as arrays
  dvar float+ d[I];
  dvar float+ r[I];

  //maximize the sum of A_i
  maximize sum(i in I) (d[i]);

  // our equations
  subject to
  {
    // By Adapted lemma 13
    forall(i in I) r[i] <= i+1;

    // By Adapted lemma 14
    forall(i in I2) r[i] <=  r[i+1]  -  ((( i-3 ) * d[i+1]) / 2);

    // By Adapted lemma 15
    forall(i in I3) d[i] <= (2 * r[i])  /  ( i-4 );

    // By Adapted lemma 16
    forall(i in I) r[i] >= sum(x in 1..i)d[x];
  }
\end{verbatim}

% \textbf{"outerplanar"}

\begin{figure}
  \begin{tikzpicture}

    \begin{axis}[
      ybar,
      xmin = 0, xmax = 15,
      ymin = 0, ymax = 9,
      xtick distance = 1,
      ytick distance = 1,
      minor tick num = 1,
      width = \textwidth,
      height = \textwidth*0.5,
      xlabel = {$i$},
      ylabel = {$d[i]$},]
    ]

    %plot the first function
    \addplot +[
      ybar,
      fill=blue,
      nodes near coords,
      nodes near coords style = {anchor=west, rotate=90}
    ] file[skip first] {results_outerplanar.txt};

    \end{axis}
  \end{tikzpicture}
  \caption{The degree distribution in outerplanar graphs}
\end{figure}
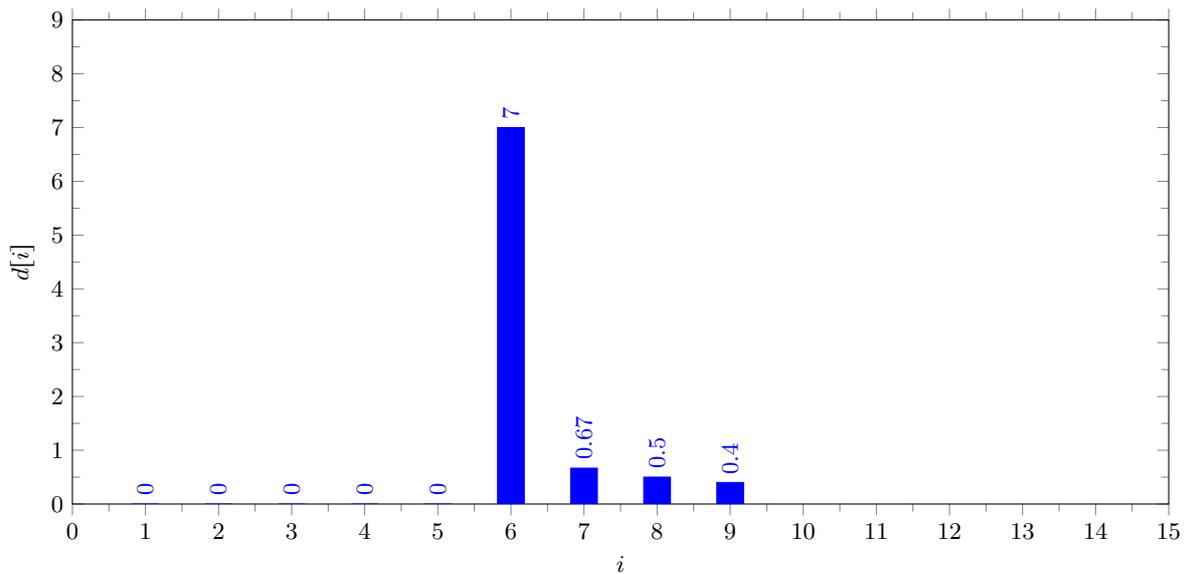

\pagebreak
\subsection{The linear program for planar graphs of girth 5}
In this case, the \cref{alem:total-H} to \ref{alem:tri-h2} can be slightly
improved, as the edge density of planar graphs of girth 5 is at most $5/3$.
This is however not so useful. As shown below, the linear constraints do not
yield something better than simply picking all $4\gamma$ non dominated vertices.

\begin{verbatim}
//girth5
  //the ranges i can have
  range I = 1..3;
  range I2 = 1..2;
  range I3 = 4..3;

  //decision variables as arrays
  dvar float+ d[I];
  dvar float+ r[I];

  //maximize the sum of A_i
  maximize sum(i in I) (d[i]);

  // our equations
  subject to
  {
    forall(i in I) r[i] >= sum(x in 1..i)d[x];

    forall(i in I) r[i] <= i+1;

    forall(i in I3) d[i] <= ((5 * r[i])  /  ( 3 * i -10 ));

    forall(i in I2) r[i] <=  r[i+1]  -  ((( 3*i-7 ) /5 )* d[i+1]);
  }
\end{verbatim}

%\textbf{"girth"}

\begin{figure}
  \begin{tikzpicture}

    \begin{axis}[
      ybar,
      xmin = 0, xmax = 15,
      ymin = 0, ymax = 9,
      xtick distance = 1,
      ytick distance = 1,
      minor tick num = 1,
      width = \textwidth,
      height = \textwidth*0.5,
      xlabel = {$i$},
      ylabel = {$d[i]$},]
    ]

    %plot the first function
      \addplot +[
      ybar,
      fill=blue,
      nodes near coords,
      nodes near coords style = {anchor=west, rotate=90}
    ] file[skip first] {results_girth.txt};

    \end{axis}
  \end{tikzpicture}
  \caption{The degree distribution in planar graphs of girth 5.}
\end{figure}
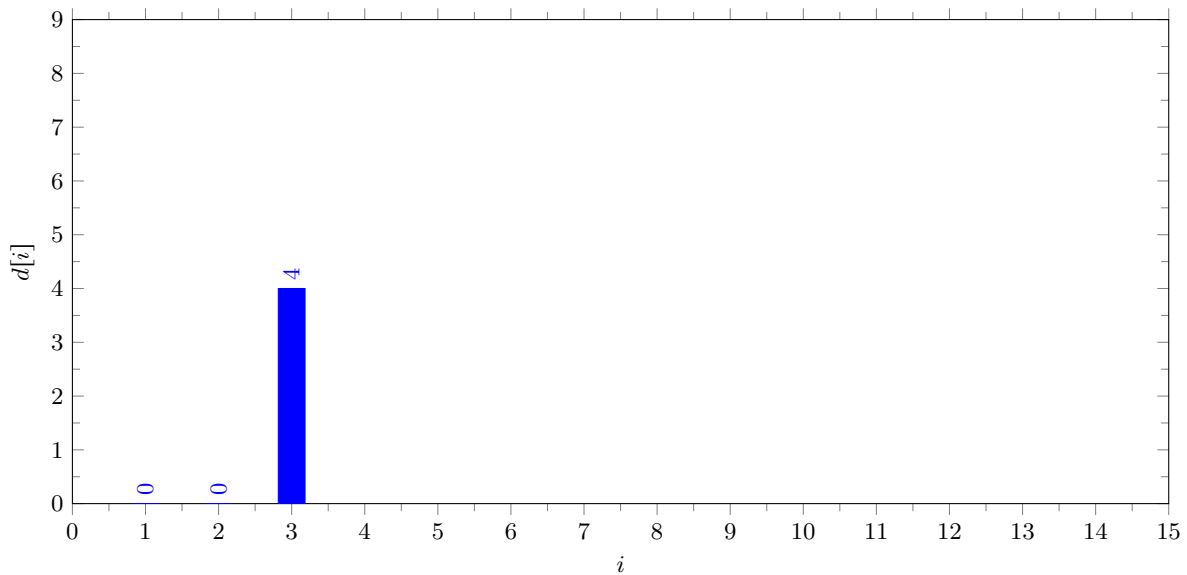